\documentclass[a4paper,11pt]{article}

\usepackage{fourier}

\usepackage{amsmath, amssymb, amsthm, amssymb,  color, courier, float, graphicx, MnSymbol, microtype, multirow, verbatim, url}
\usepackage[normalem]{ulem}

\usepackage{xcolor}
\usepackage[colorlinks = true,
            linkcolor = red,   	
            urlcolor  = blue, 	
            citecolor = blue, 	
            anchorcolor = blue 	
            ]{hyperref}

\usepackage{authblk} 
\usepackage{rotating} 

\usepackage{tikz}
\usepackage{caption}

\usepackage{algorithm2e}

\newtheorem{theorem}{Theorem}
\newtheorem{lemma}[theorem]{Lemma}
\newtheorem{corollary}[theorem]{Corollary}
\newtheorem{conjecture}[theorem]{Conjecture}
\newtheorem{definition}[theorem]{Definition}

\newtheorem{proposition}[theorem]{Proposition}
\newtheorem{question}[theorem]{Question}

\newtheorem{claim}[theorem]{Claim}
\numberwithin{theorem}{section}

\makeatletter
\newtheorem*{rep@theorem}{\rep@title}
\newcommand{\newreptheorem}[2]{%
\newenvironment{rep#1}[1]{%
 \def\rep@title{#2 \ref{##1}}%
 \begin{rep@theorem}}%
 {\end{rep@theorem}}}
\makeatother
\newreptheorem{theorem}{Theorem}
\newreptheorem{lemma}{Lemma}
\newreptheorem{corollary}{Corollary}
\newreptheorem{conjecture}{Conjecture}
\newreptheorem{definition}{Definition}
\newreptheorem{observation}{Observation}
\newreptheorem{problem}{Problem}
\newreptheorem{proposition}{Proposition}
\newreptheorem{question}{Question}
\newreptheorem{remark}{Remark}

\newcommand \minc{min\text{-}cost}
\newcommand \maxc{max\text{-}cost}

\newcommand\red[1]{{\color{red}#1}}
\newcommand\blue[1]{{\color{blue}#1}}

\title{Graph properties in node-query setting:  effect of breaking symmetry}

\author[1]{Nikhil Balaji 
}
\author[2]{Samir Datta 
}
\author[3]{Raghav Kulkarni 
}
\author[4]{\\Supartha Podder 
}
\affil[1,2]{\footnotesize Chennai Mathematical Institute,  India\\ \{nikhil, sdatta\}@cmi.ac.in}
\affil[3,4]{\footnotesize Centre for Quantum Technologies, Singapore\\ \{kulraghav, supartha\}@gmail.com}

\begin{document}
\date{}
\maketitle

\begin{abstract}
The query complexity of graph properties is well-studied when queries are on edges.
We investigate the same when queries are on nodes. 
In this setting a graph  $G = (V, E)$ on $n$ vertices
and a property $\mathcal{P}$  are given. 
A black-box access to an unknown subset $S \subseteq V$ is 
provided via queries of the form `Does $i \in S$?'.
 We are interested in the minimum
number of queries needed in worst case in order to determine whether $G[S]$ -- the subgraph of $G$ induced on $S$ --
satisfies $\mathcal{P}$.

Apart from being combinatorially rich, this setting appears to be a natural abstraction of several scenarios in areas including
computer networks and social networks where one is interested in properties of the underlying sub-network on a set of active nodes.

Another reason why we found this setting interesting is because it allows us to initiate
a systematic study of breaking symmetry in the context of query complexity of graph properties. In particular,
we focus on hereditary graph properties -- the properties closed under deletion of vertices as well as edges.
The famous Evasiveness Conjecture asserts that
even with a minimal symmetry assumption on $G$, namely that of vertex-transitivity, the query complexity for any hereditary graph property in our setting is the worst
possible, i.e., $n$. 

We show that in the absence of any symmetry on $G$ it can fall as low
as $O(n^{1/(d + 1) })$ where $d$ denotes the minimum possible degree of a minimal forbidden sub-graph for $\mathcal{P}$. In particular, every hereditary property benefits at least quadratically. 
The main question left  open is: 
can it go exponentially low for some hereditary property?
 We show that the answer is no for any hereditary property with {finitely many} forbidden subgraphs by exhibiting a bound of $\Omega(n^{1/k})$ for some constant $k$ depending only on the property. For general ones we rule out the possibility of the query complexity falling down to constant by showing $\Omega(\log n/ \log \log n)$ bound. Interestingly, our lower bound proofs rely on the famous Sunflower Lemma due to Erd\"os and Rado.


\end{abstract}

\pagebreak
\section{Introduction}
\label{sec:intro}

\subsection{The query model}
The decision tree model (aka query model), perhaps due to its  simplicity and fundamental nature, has been extensively studied in the past
and still remains a rich source of many fascinating investigations.
In this paper we focus on Boolean functions, i.e., the functions of the form $f : \{0,1\}^n \to \{0, 1\}.$
 A deterministic decision tree $D_f$ for $f$
takes $x = (x_1, \ldots, x_n)$ as an input and determines the value
of $f(x_1, \ldots, x_n)$ using queries of the form 
$\text{``is } x_i = 1? \text{"}$.  Let $C (D_f, x)$ denote the cost of the computation, that is the number of queries made by $D_f$ on
input $x.$ The {\em deterministic decision tree complexity} (aka deterministic query complexity) of $f$ is defined as
\[D(f) = \mathop{\min}_{D_f} \max_x C (D_f, x). \]

We encourage the reader to see an excellent survey by Buhrman and de Wolf \cite{bw} 
on decision tree complexity. 
We also note that the randomized and the quantum variants \cite{bw} of decision trees have also been extensively studied
in the past. Several different variants such as parity decision trees have been studied in connection to communication complexity, learning,
and property testing \cite{sz, km, bck}.

\subsection*{Why are query models important?}
Variants of the decision tree model are fundamental for several reasons.
Firstly they occur naturally in connection to the other models of computation
such as communication complexity\cite{sz}, property testing\cite{bck}, learning\cite{km}, circuit complexity\cite{imp} etc.
Secondly decision tree models are much simpler to analyse compared to 
other models such as circuits. Thus one can actually hope to use them as 
a tool in the study of other models.
Thirdly these models are mathematically rich and beautiful. Several connections to algebra, combinatorics, topology, Fourier analysis, and number theory  \cite{odo, bbkn} make the decision tree models interesting on their own right.
Finally there remain some fascinating open questions \cite{kul-survey} in query complexity 
that have attracted the attention of generations of researchers
over the last few decades by their sheer elegance and notoriety.

\subsection{Graph properties in node-query setting}
In this paper 
we investigate the query complexity of graph properties. In particular we focus on
the following setting:
A graph  $G = (V, E)$
and a property $\mathcal{P}$  are fixed. 
We have access to $S \subseteq V$ via queries of the form ``Does $i \in S$?".
 We are interested in the minimum
number of queries needed in the worst case in order to determine whether $G[S]$ -- the subgraph of $G$ induced on $S$ --
satisfies $\mathcal{P}$, which we denote by $cost(\mathcal{P}, G)$. One may define a similar notion of cost 
for randomized and quantum 
models.

We call $G$ the base graph for $\mathcal{P}$. We say that a vertex $i$ of $G$ is relevant for $\mathcal{P}$ if there exists
some $S$ containing $i$ such that exactly one of $G[S]$  and $G[S - \{i\}]$ satisfies $\mathcal{P}$. 
 We say that $G$ is relevant for $\mathcal{P}$ if all its vertices are relevant for $\mathcal{P}$.
The minimum possible cost of  $\mathcal{P}$, denoted by\footnote{We slightly abuse this notation by omitting the subscript $n$.}
 $min\text{-}cost(\mathcal{P})$, 
is defined as follows: 
\[min\text{-}cost_n(\mathcal{P}) = \min_G\{cost(\mathcal{P}, G) \mid  \text {$G$ is relevant for $\mathcal{P}$ } \& \ |V(G)| = n\}. \]
Similarly one can define $max\text{-}cost(\mathcal{P})$ as follows:
\[max\text{-}cost_n(\mathcal{P}) = \max_G\{cost(\mathcal{P}, G) \mid  \text {$G$ is relevant for $\mathcal{P}$ } \& \ |V(G)| = n\}. \]
The $max\text{-}cost$ is a more natural notion of complexity when one is interested in studying the universal upper bounds.
Investigating the  $max\text{-}cost$ in our setting can indeed be a topic of an independent interest.
However, for the purpose of this paper, the notion of $min\text{-}cost$ will be more relevant as we are interested in finding how low 
can the universal lower bound on query complexity go under broken symmetry
(Refer to Section~\ref{sec:eff-sym-break} for more on symmetry).
  It turns out that in the presence of symmetry this bound  is $\Omega(n)$
for most of the properties and it is conjectured to be $\Omega(n)$ for any hereditary property in our setting.
Recall that  a hereditary property is a property of graphs, which is closed under deletion of vertices as well as edges. For instance 
acyclicity, bipartiteness, planarity, and containing
a triangle are hereditary properties whereas connectedness and containing a perfect matching are not.  Every hereditary property can be described by a 
(not necessarily finite) 
collection  of its 
forbidden subgraphs.\footnote{In our setting, every hereditary property is a monotone Boolean function.}
\footnote{We would like to highlight that although we didn't explicitly define $min\text{-}cost(\mathcal{P})$ or $max\text{-}cost(\mathcal{P})$ for randomized query model, all our lower bound proofs are based on sensitivity arguments and hence work even for randomized case.}

\subsection*{Why node-query setting might be interesting?}
Below we illustrate with some examples why it might be interesting to investigate
several complexity measures of a graph property in the node-query setting.
\subsubsection*{Example 1}
Consider a graph that models the associations in a social network, say the Facebook graph (where two nodes are adjacent if they are Facebook friends). At any given time, users can be online or offline. We might be interested in finding out if there is any user who is online and is influential,
in the sense that he/she has many neighbors (friends) who are also online at that time.
This problem can be formulated in our setting as whether the induced subgraph has a vertex of large degree or not. 
(Appendix~\ref{sec:bounded-degree})
\subsubsection*{Example 2}
Consider a physical network with several nodes and links between them. At any given time, the nodes of the network can be either
active or inactive. One way to find out if a node is active is to ping it (possibly by physically going to the site), which comes with some fixed cost.
For example, the underlying network could be the network of routers which are physically connected by wires. Some of the routers
may go on and off over time. At any given time, we want to know whether a message can be sent between two specified nodes
via the active routers. 
This problem can be formulated in our setting as whether the subgraph induced by active routers has a path between two specified 
nodes $s$ and $t$ or not. (Appendix~\ref{connectivity})

\subsubsection*{Example 3}
Consider a chemical lab which performs experiments with certain basic ingredients to build medicines. 
Suppose a concoction comes out of an experiment and one wants to know whether it is harmful or not. There are tests available for testing the presence of an ingredient in the concoction. The lab also has a table of which two ingredients together form a harmful combination. So the 
goal is to perform the tests for presence of individual ingredients to check if any of the harmful combination is present.
This problem can be formulated in our setting as whether the induced subgraph is empty or not.
(Appendix~\ref{sec:ind})

It appears that our setting is a natural abstraction of these  type of scenarios, where 
 one is interested in  the properties of subgraph induced
by active nodes in a network. 
To the best of our knowledge, no systematic study of node-query setting has been yet undertaken.  
Here we initiates such a line of inquiry for graph properties. In particular,  we focus on the role of presence and absence of {\em symmetry}.

\subsection{Effect of breaking symmetry}
\label{sec:eff-sym-break}
Another reason why our setting is interesting is that it allows us to study the effect of breaking symmetry on query complexities of graph properties. In particular, our setting provides a platform to compare
 the complexity of $\mathcal{P}$ when the base
graph $G$ has certain amount of symmetry with the complexity of $\mathcal{P}$ when $G$ has no symmetry whatsoever.
To formalize this, we define the notion of $\mathcal{G}\text{-}min\text{-}cost(\mathcal{P})$
for a class of graphs $\mathcal{G}$ by restricting ourselves only to graphs in $\mathcal{G}$.
\[\mathcal{G}\text{-}min\text{-}cost_n(\mathcal{P}) = \min_{G \in \mathcal{G}} \{cost(\mathcal{P}, G) \mid  \text {$G$ is relevant for $\mathcal{P}$ } \& \ |V(G)| = n\}. \]

When $\mathcal{G}$ has the highest amount of symmetry, i.e., when $\mathcal{G}$ is the class of complete graphs, then it is easy to see that for every hereditary $\mathcal{P}$, $\mathcal{G}\text{-}min\text{-}cost(\mathcal{P})$ is nearly the worst possible
,  i.e., $\Omega(n)$.  It  turns out that one does not require the whole symmetry of the complete graph
to guarantee the $\Omega(n)$ bound. Even weaker symmetry assumptions on graphs in $\mathcal{G}$, for instance being Caley graphs
of some group, indeed suffices. Thus it is natural to ask how much symmetry 
is required to guarantee  the $\Omega(n)$ bound.  
In fact, the famous Evasiveness Conjecture implies that even under the weakest form of symmetry on $\mathcal{G}$, i.e.,  when $\mathcal{G}$ is the class
of transitive graphs, for any hereditary property $\mathcal{P}$  the $\mathcal{G}\text{-}min\text{-}cost(\mathcal{P})$ would remain the highest possible, i.e., $n$.  
So for the complexity
to fall down substantially 
we might have to let go of the transitivity of $\mathcal{G}$. This is exactly what we do. In particular we take $\mathcal{G}$ to be the class of all graphs, i.e., we assume no
symmetry whatsoever. Note that in this case  $\mathcal{G}\text{-}min\text{-}cost(\mathcal{P}) = min\text{-}cost(\mathcal{P})$
that we defined earlier.  
Now a natural question is how low can
$min\text{-}cost(\mathcal{P})$ go in the absence of any symmetry? This is the main question addressed by our paper.
In particular, we show that  for any hereditary property $\mathcal{P}$, the $min\text{-}cost(\mathcal{P})$ falls down
at least quadratically, i.e, to $O(\sqrt n)$.  For some properties, it can go even further below (polynomially down) with polynomials of arbitrary constant degree,
i.e. to $O(n^{1/k})$ where $k$ is a constant depending only on the property. 
The main question left  open by our work is: does there exist a hereditary property $\mathcal{P}$ for which
 $min\text{-}cost(\mathcal{P})$ is exponentially low?  In other words:
\begin{question}
Is it true that for every hereditary property  $\mathcal{P}$ there exists an integer $k_\mathcal{P}>0$ such that
 \[ min\text{-}cost(\mathcal{P}) = \Omega(n^{1/k_\mathcal{P}})?\]
 \label{question:exponential}
\end{question}

\subsection*{Related work}
Understanding the effect of symmetry on computation is a very well-studied theme in the past. 
Perhaps its roots can also be traced back to 
the non-solvability of quintic equations by radicals -- the legendary work of 
Galois \cite{wiki}. In the context of query complexity, again there has been a 
substantial amount of effort invested in understanding
the role of symmetry. A recurrent theme here is to exploit the symmetry and some other structure \cite{ks} of the underlying functions to prove good lower bounds
on their query complexity. For instance the famous Andera-Rosenberge-Karp Conjecture \cite{kss} asserts that every non-trivial monotone graph property of $n$ vertex graphs (in the edge-query model) must  be evasive, i.e., its query complexity is $n \choose 2$.  While a weaker bound of  $\Omega(n^2)$ is known,
the conjecture remains widely open to this date. Several special cases 
of the conjecture have also been studied \cite{cks}. The randomized query complexity of monotone graph properties is also conjectured to be $\Omega(n^2)$ \cite{fsw}. The generalizations
of these conjectures for arbitrary transitive Boolean functions are also studied: In particular, recently Kulkarni  \cite{kul} has formulated the Weak-Evasiveness Conjecture for monotone transitive functions, which vastly generalize monotone graph properties. In the past, 
Lov\'  asz had conjectured \cite{pclg} the evasiveness of checking independence of $S$ exactly in our setting. Sun,Yao, and Zhang
\cite{syz}
 study query complexity of graph properties and several transitive functions including the circulant ones.
Their motivation was to investigate how low can the query complexity go if one drops the assumption of monotonicity or lower the amount of symmetry. In this paper, we follow their footsteps and ask the same question under no symmetry assumption whatsoever. The main difference between the past
works and this one is that most of the previous work exploit the symmetry to prove (or to conjecture) a good lower bound, whereas we investigate the consequences 
of breaking the symmetry for the query complexity. 

\subsection{Our main results} 
In this section we summarize our main results. 
Let $\mathcal{P}$ be a hereditary graph property
and $d_\mathcal{P}$ denote the minimum possible degree
of a minimal forbidden subgraph for $\mathcal{P}$.

\begin{theorem}
For any  hereditary graph property $\mathcal{P}$: \[min\text{-}cost(\mathcal{P}) = O(n^{1/(d_\mathcal{P} + 1)}).\] \label{thm:upper-d}
\label{thm:upper}
\end{theorem}
 
\begin{corollary}
\label{cor:upper-d}
For any hereditary graph property $\mathcal{P}$:  \[min\text{-}cost(\mathcal{P}) = O(\sqrt n).\] 
\end{corollary}
We note that the above upper bound does not hold for general graph properties. For instance Connectivity has cost $\Theta(n)$, so does containment of a Perfect Matching, which are both non-hereditary properties
(See Appendix~\ref{connectivity}).

As a partial answer to Question~\ref{question:exponential} 
we prove the following theorem.

\begin{theorem}
\label{thm:forbidden-H}
\label{containing-H}
Let $H$ be a fixed graph on $k$ vertices and let $\mathcal{P}_H$ denote the property of containing $H$ as a subgraph. Then,
\[ \minc(\mathcal{P}_H) = \Omega(n^{1/k}).\]
\end{theorem}

Interestingly our proof of Theorem \ref{thm:forbidden-H} uses the famous Sunflower Lemma due to Erd\"os and Rado~\cite{er}. Moreover it generalizes to any fixed number of forbidden subgraphs each on at most $k$ vertices.

We note that both Theorem \ref{thm:upper} and Theorem~\ref{thm:forbidden-H} are not tight. However, we do prove tight bounds for several hereditary properties. We summarize few such interesting bounds in the Table below.

\addtocounter{footnote}{0} 
 \stepcounter{footnote}\footnotetext{assuming Weak Evasiveness.}
 \stepcounter{footnote}\footnotetext{when $d(G) \geq 7$.}
\addtocounter{footnote}{-2} 

\begin{table}[H]
\resizebox{1.2\textwidth}{!}{
\begin{tabular}{|c|c c|c|c|}
\hline 
& \multicolumn{2}{|c|}{Properties} & With Symmetry\footnote{assuming Weak Evasiveness} & Without Symmetry \\
 \hline 
 \hline
 
\multirow{10}{*}{\begin{turn}{90}\ \ \ \ \ \ \ \ \ Finite\end{turn}} & Independence/Emptiness   & [Thm.~\ref{thm:ind-lower}]  & $\Theta(n)$ & $\Theta(\sqrt{n})$\\
\cline{2-5}

& Bounded Degree   & [Thm.~\ref{thm:bounded-degree}]  & $\Theta(n)$ &  $\Theta(\sqrt{n})$\\
\cline{2-5}

 & Triangle-freeness   & [Thm.~\ref{thm:triangle-free}]  & $\Theta(n)$ & $\Theta(n^{1/3})$ \\
\cline{2-5}

 & Containing $K_t$   & [Thm.~\ref{thm:upper-d}][Thm.~\ref{containing-H}]  & $\Theta(n)$ & $\Theta(n^{1/t})$\\
\cline{2-5}

 & Containing $P_t$   & [Thm.~\ref{thm:upper-d}][Thm.~\ref{containing-H}]  & $\Theta(n)$ & $O(\sqrt n)$,$\Omega(n^{1/t})$ \\
\cline{2-5}

 & Containing $C_t$   & [Thm.~\ref{thm:upper-d}][Thm.~\ref{containing-H}]  & $\Theta(n)$ & $O(n^{1/3})$, $\Omega(n^{1/t})$\\
\cline{2-5}

 & Containing $H$: $V(H)=k$    &  [Thm.~\ref{thm:H-symmetry}][Thm.~\ref{thm:upper-d}][Thm.~\ref{containing-H}] & $\Theta(n)$ & 
 $O(n^{1/(d_{min}+1)})$, $\Omega(n^{1/k})$ \\
\hline
\hline

\multirow{4}{*}{\begin{turn}{90}Infinite\end{turn}} & Acyclicity   & [Thm.~\ref{thm:acyc-trans-lower}]  &  $\Theta(n)$ & $O(n^{1/3})$
\\
\cline{2-5}

 & Bi-partiteness   & [Thm.~\ref{thm:upper-d}]  & Open   & $O(n^{1/3})$\\
\cline{2-5}

 & 3-colorability   & [Thm.~\ref{thm:upper-d}]  & Open  & $O(n^{1/4})$\\
\cline{2-5}

& Planarity   & [Thm.~\ref{thm:appx:symmetry-planar}]  & $\Theta(n)$\footnote{when $d(G) \geq 7$} &  $O(n^{1/4})$\\

\hline

\end{tabular}
}
 \label{table-results}
 \caption{Summary of Results for Finite/Infinite Forbidden Subgraphs}
\end{table}

Finally we note a non-constant lower bound, which holds for {\em any} hereditary property. Our proof again
relies on the Sunflower Lemma.
\begin{theorem} 
For any hereditary graph property $\mathcal{P}$
\[ \minc(\mathcal{P}) = \Omega \left( \frac{\log n}{\log \log n}\right).\]
\label{thm:lower}
\label{thm:any-her-lower}
\end{theorem}

As we use sensitivity argument all our lower bounds work for randomized case as well.

\subsection*{Organization}   
Section~\ref{sec:prelims} covers some preliminaries. 
Section~\ref{sec:symmetry} contains Weak Evasiveness results under symmetry. Section~\ref{sec:symmetry-breaking} contains proofs of Theorem~\ref{thm:upper} and Theorem~\ref{thm:forbidden-H}.
Proof of Theorem~\ref{thm:lower} and proof of some tight bounds for Theorem~\ref{thm:upper} are deferred to Appendix~\ref{appx:her-lower} and Appendix~\ref{appx:tight-bound} respectively. 
Next in Section~\ref{sec:restrictions} we state some results on restricted graph classes. 
Proofs of these results are in Appendix~\ref{appx:restrict}. 
Finally in Section~\ref{sec:open} we discuss some open directions.

\section{Preliminaries}\label{sec:prelims}
Let $[n] := \{1, \ldots, n\}$.
\subsection*{Randomized query complexity}
A randomized decision tree $\mathcal{T}$ is simply a probability distribution on the deterministic decision trees $\{T_1, T_2, \ldots\}$ where the tree $T_i$ occurs with
probability $p_i.$ We say that $\mathcal{T}$ computes $f$ correctly if for every input
$x$: $\Pr_{i}[T_i(x) = f(x)] \geq 2/3$. The depth of $\mathcal{T}$ is the maximum depth
of a $T_i$. The (bounded error) randomized query complexity of $f$, denoted by $R(f)$,  is the minimum possible
depth of a randomized tree computing $f$ correctly on all inputs.

\subsection*{Some classes of Boolean function}
A Boolean function $f:\{0,1\}^n \to \{0,1\}$ is said to be {\em monotone}
increasing if for any $x \leq y$, we have $f(x) \leq f(y)$, where $x \leq y$ means
$x_i \leq y_i$ for all $i \in [n]$. Similarly one can define a monotone decreasing function. A Boolean function $f(x_1, \ldots, x_n)$ is said to be
{\em transitive} if there exists a group $G$ that acts transitively on the variables $x_i$s such that $f$ is invariant under this action, i.e., 
for every $\sigma \in G$: $f(x_{\sigma_1}, \ldots, f_{\sigma_n}) =
f(x_1, \ldots, x_n)$. A Boolean function $f:\{0,1\}^n \to \{0, 1\}$ 
is said to be {\em evasive} if $D(f) = n$.

\subsection*{Hereditary graph properties}
A property $\mathcal{P}$ of graphs is simply a collection of graphs.
The members of $\mathcal{P}$ are said to satisfy $\mathcal{P}$ and non-members are said to fail $\mathcal{P}$. A property is hereditary if it is
closed under deletion of vertices as well as edges\footnote{vertex-hereditary: closed only under vertex-deletion (e.g. being chordal).}. For instance: acyclicity, planarity,
and $3$-colorability are hereditary properties, whereas connectivity and containing
a perfect matching are not. Every hereditary property $\mathcal{P}$ can be uniquely expressed as a (possibly infinite) family $\mathcal{F}_\mathcal{P}$ of its forbidden subgraphs. For instance: acyclicity can be described as forbidding all cycles. 
Given a graph $G$, the hitting set $S_{G,\mathcal{P}}$ for $\mathcal{P}$ is a subset of $V(G)$ such that removing $S_{G,\mathcal{P}}$ from $G$ would make the property $\mathcal{P}$ absent\footnote{such that every graph in $\mathcal{F}_{\mathcal{P}}$ shares a node with $S_{\mathcal{G}, \mathcal{P}}$.}.
Hereditary graph properties in node-query setting are monotone decreasing Boolean functions.
Sometimes we refer hereditary properties by their negation. For instance: containing triangle.

\subsection*{Sensitivity and block-sensitivity \cite{hkp}}
The $i^{th}$ bit of an input $x \in \{0,1\}^n$ is said to be sensitive for $f:\{0,1\}^n \to \{0, 1\}$ if
$f(x_1, \ldots, x_i, \ldots, x_n) \neq f(x_1, \ldots, 1 - x_i, \ldots, x_n)$.
The sensitivity of $f$ on $x$, denoted by $s_{f,x}$ is the total number of sensitive
bits of $x$ for $f$. The sensitivity of $f$, denoted by $s(f)$, is the maximum
of $s_{f,x}$ over all possible choices of $x$. A block $B \subseteq [n]$ of variables
is said to be sensitive for $f$ on input $x$, if flipping the values of all $x_i$ such that $i \in B$ and keeping the remaining $x_i$ the same,  results in flipping the output of $f$. The block sensitivity of $f$ on an input $x$,
denoted by $bs_{f,x}$ is the maximum number of {\em disjoint} sensitive blocks for $f$ on $x$. The block sensitivity of a function $f$, denoted by $bs(f)$, is the maximum value of $bs_{f,x}$ over all possible choices of $x$. It is known that
$D(f) \geq R(f) \geq bs(f) \geq s(f)$. For monotone functions, $bs(f) = s(f)$.
\section{Presence of symmetry in node-query setting: \\ does it guarantee weak-evasiveness?}
\label{sec:symmetry}
In edge-query setting, Aanderaa-Rosenberg-Karp Conjecture \cite{kss, cks} asserts that any non-trivial monotone graph property must be evasive, i.e., one must query
all $n \choose 2$ edges in worst-case. The following generalization of the ARK Conjecture asserts that only monotonicity and modest amount of symmetry, namely transitivity, suffices to guarantee
the evasiveness \cite{lutz}.
\begin{conjecture}
[Evasiveness Conjecture] Any non-constant monotone transitive function $f$ on $n$ variables must have $D(f) = n$.
\end{conjecture}

This conjecture appears to be notoriously hard to prove even in several interesting special cases. Recently Kulkarni \cite{kul} formulates: 
\begin{conjecture}[Weak Evasiveness Conjecture]
If $f_n$ is a sequence of monotone transitive functions on $n$ variables then for every $\epsilon > 0$:
 \[D(f_n) = \Omega(n^{1-\epsilon}).\]
\end{conjecture}

Although Weak EC appears to be seemingly weaker, Kulkarni \cite{kul} observes that it is equivalent to the EC itself. His results hint towards the possibility that
disproving Weak EC might be as difficult as separating $TC^0$ from $NC^1$. However: proving special cases of Weak EC appears to be relatively 
less difficult. In fact, Rivest and Vuillemin \cite{rv} confirm the Weak EC  for graph properties and recently Kulkarni, Qiao, and Sun \cite{kqs} confirm Weak EC for 3-uniform hyper graphs and Black \cite{black} extends this result to $k$-uniform hyper graphs.
All these results are studied in the edge-query setting. It is natural to ask whether the Weak EC becomes tractable in node-query setting.
The monotone functions in node-query setting translate precisely to the hereditary graph properties.
Here we show that it does become tractable for several hereditary graph properties.
But first we need the following lemma \cite{sc05, syz}:

\begin{lemma} Let $f$ be a monotone transitive function.  Let $k$ be the size of a 1-input with minimal number of 1s. Then:
$D(f) = \Omega(n/k^2)$.
\label{lem:mon-trans}
\end{lemma}

Let 
$\mathcal{G}_{\mathcal{T}}$
denote the class of transitive graphs.  Let $H$ be  a fixed graph.  Let $\mathcal{P}_H$ denote the property of containing $H$ as a subgraph.
The following theorem directly follows from Lemma~\ref{lem:mon-trans}.
\begin{theorem}
\label{thm:H-symmetry}
\[\mathcal{G}_{\mathcal{T}}\text{-}\minc(\mathcal{P}_H) = \Omega(n). \]

\end{theorem}
The above result can be generalized for any finite family of forbidden subgraphs. We do not yet know how to prove it for infinite family in general.
However below we illustrate a proof for one specific case when infinite family is the family of cycles. First we need the following lemma:

\begin{lemma}
Let $G$ be a graph on $n$ vertices,  $m$ edges,  and maximum degree $d_{max}$. Let $\mathcal{C}$ denote the property of being acyclic. Then,
\[cost(\mathcal{C}, G) \geq (m-n)/d_{max}.\]
\label{lem:d-max}
\end{lemma}
\begin{proof} To make $G$ acyclic one must remove at least $m-n$ edges. Removing one vertex can remove at most $d_{max}$ edges.
Thus the size of minimum feedback vertex set (fvs)  is at least $(m-n)/d_{\max}$. The adversary answers all vertices outside this fvs
to be present. Now the algorithm must query every vertex in the minimum fvs.
\end{proof}

\begin{theorem} 
\label{thm:acyc-trans-lower}
 \[ \mathcal{G}_{\mathcal{T}}\text{-}\minc(\mathcal{C}) = \Omega(n).\]

\end{theorem}
\begin{proof}
Since $G$ is transitive, $G$ is $d$ regular for some $d$ \cite{agt}. Therefore $m = dn/2$ and $d_{max} = d$. Hence from Lemma~\ref{lem:d-max} we get the desired bound.
\end{proof}

We also show similar bound for the property of being planar (See Appendix~\ref{appx:symmetry}).
 Following special case of Weak EC remains open:
\begin{conjecture} For any hereditary property $\mathcal{P}$, for any $\epsilon >0$:
\[\mathcal{G}_{\mathcal{T}}\text{-}\minc(\mathcal{P}) = \Omega(n^{1-\epsilon}).\]

\end{conjecture}

\section{Absence of symmetry in node-query setting: \\ how low can query complexity go?}
\label{sec:symmetry-breaking}

\subsection{A general upper bound}
\label{sec:upper}

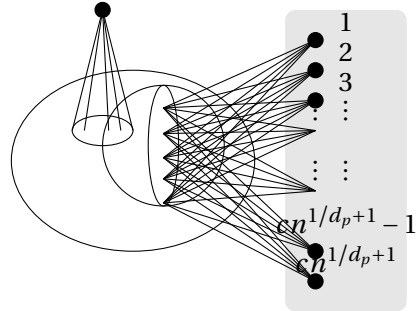
\begin{figure}[H]
 \begin{center}

\begin{tikzpicture}[scale=0.4]

 \tikzstyle{nnode} = [circle, draw=black, fill=black, inner sep=0pt, minimum size = 2mm]
 \tikzstyle{txt} = [above, text width=2cm, text centered]

\fill [rounded corners, gray!20] (5,-5) rectangle (9,5);

\draw (0,0) ellipse (4cm and 3cm);
\draw (-1,1) ellipse (1cm and 0.5cm);
\draw (1,0.5) circle (2cm);
\draw (1,0.5) ellipse (0.5cm and 2cm);

\node (u) [nnode] at (-1,5){};

\foreach \y in {4,3,2,-3,-4}
\node () [nnode] at (6,\y){};
\node () [txt] at (6,-1){$\vdots$};
\node () [txt] at (6,1){$\vdots$};

\foreach \y/\ytext in {4/1,3/2,2/3,1/\vdots,-1/\vdots, -2.7/cn^{1/d_p+1}-1, -4/cn^{1/d_p+1}}{
\node () [txt] at (7,\y){$\ytext$};
};

\foreach \x in {-2,-1.6,-0.8,-0.4,0}{
\node () [txt] at (\x,1){};
\draw (u)--(\x,1);
};

\foreach \x in {1.75,0.9,0.1,-0.6,-1.4}
\foreach \y in {4,3,2,1,-1,-3,-4}
{
\node () [txt] at (1,\x){};
\draw (6,\y)--(1,\x);
};

 \end{tikzpicture}

 \end{center}
\caption{
 Construction of $G$ for a general upper bound  
}\label{fig:upper-bound}
\label{fig:general}
\end{figure}

Let $\mathcal{P}$ be a hereditary graph property
and $d_\mathcal{P}$ denote the minimum possible degree
of a minimal forbidden subgraph for $\mathcal{P}$.

{\sf Proof of Theorem~\ref{thm:upper}:} 
Let $k = c \cdot n^{1/(d_\mathcal{P}+1)}$ where we choose the constant $c$ appropriately.
Construct a graph $G$ on $n$ vertices as follows (See Figure~\ref{fig:upper-bound}):
\begin{itemize}
\item Start with a clique on vertices $v_1, \ldots, v_k$.
\item For every $S \subseteq [k]$ such that $|S| = d_\mathcal{P}$
 \begin{itemize}
\item add $k$ new vertices $u_1^S, \ldots, u_k^S$ and 
\item connect each $u_i^S$ to every $v_i : i \in S$.
\end{itemize}
\end{itemize}

Note that every vertex of $G$ is relevant for $\mathcal{P}$. Now we describe an algorithm (See Algorithm~\ref{alg}) to  determine $\mathcal{P}$ in $O(n^{1/(d_\mathcal{P} +1)})$ queries.
Let $c_\mathcal{P}$ denote the smallest integer such that the clique on $c_\mathcal{P}$ vertices satisfies $\mathcal{P}$.

\RestyleAlgo{boxruled}
\LinesNumbered
\begin{algorithm}[ht]
  \caption{ \label{alg}}
  
 \begin{itemize}
\item Query $v_1, \ldots, v_k$.
\item If at least $c_\mathcal{P}$ of these vertices are present then $\mathcal{P}$ must fail.
\item Otherwise there are at most $c_\mathcal{P}-1$ vertices present\\ (wlog: $v_1, \ldots, v_{c_\mathcal{P}-1}$).
\begin{itemize}
\item For every subset $S \subseteq [c_\mathcal{P}-1]$ such that $|S| = d_\mathcal{P}$, query $u_1^S, \ldots, u_k^S$.
\item If the graph induced on the nodes present (after all these\\ ${{c_\mathcal{P}-1} \choose d_\mathcal{P}} \times k$ queries) satisfies 
 $\mathcal{P}$ then answer Yes.\\ Otherwise answer No.
\end{itemize}
\end{itemize}
\end{algorithm}

\subsubsection*{Correctness and complexity}
Note that any vertex that is not queried by the above algorithm can have at most $d_\mathcal{P} - 1$ edges to the vertices in the clique $v_1, \ldots, v_k$.
Since $d_\mathcal{P}$ is the minimum degree of a minimal forbidden subgraph for $\mathcal{P}$, these vertices now become irrelevant for $\mathcal{P}$.
Thus the algorithm can correctly declare the answer based on only the queries it has made.
It is easy to check that the query complexity of the above algorithm is $O(k)$ which is $O(n^{1/(d_\mathcal{P}+1)})$.

\hfill $\Box$

 This completes the proof of Theorem~\ref{thm:upper-d}. Corollary~\ref{cor:upper-d} follows from this by observing that $d_\mathcal{P} \geq 2$ for any non-trivial $\mathcal{P}$.

\subsection{A general lower bound}
\label{sec:lower1}
In this section we prove Theorem~ \ref{thm:forbidden-H}.

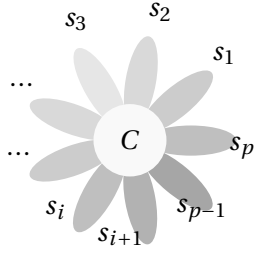
\begin{figure}[H]
 \centering
    \begin{tikzpicture}[scale=0.4]
    \tikzstyle{rect}=[rectangle, thick, 
    draw=black, rounded corners, node distance = 2cm]
    \tikzstyle{circ}=[circle, thick, 
    draw=black, rounded corners, node distance = 1cm,
    font=\tiny]
    \tikzstyle{txt} = [above, text width=2cm, text centered]

    \fill[gray!20, rotate = 30] (0.01,2) ellipse (0.5cm and 1.5cm);
    \fill[gray!30, rotate = 70] (0.01,2) ellipse (0.5cm and 1.5cm);
    \fill[gray!40, rotate = 110] (0.01,2) ellipse (0.5cm and 1.5cm);
    \fill[gray!50, rotate = 150] (0.01,2) ellipse (0.5cm and 1.5cm);
    \fill[gray!60, rotate = 190] (0.01,2) ellipse (0.5cm and 1.5cm);
    \fill[gray!70, rotate = 230] (0.01,2) ellipse (0.5cm and 1.5cm);
    \fill[gray!50, rotate = 270] (0.01,2) ellipse (0.5cm and 1.5cm);
    \fill[gray!40, rotate = 310] (0.01,2) ellipse (0.5cm and 1.5cm);
    \fill[gray!30, rotate = 350] (0.01,2) ellipse (0.5cm and 1.5cm);
    
   \fill[gray!5] (0,0) circle (1.2cm);
    \node at (0,0){$C$};

    \foreach \x/\xtext in {36/s_1, 75/s_2, 118/s_3, 160/\cdots, 195/\cdots, 230/s_i, 265/s_{i+1}, 308/s_{p-1}, 346/s_p}
    \node [txt] at (\x:3.8){$\xtext$}; 
\end{tikzpicture}
   
\end{figure}

\begin{definition}[Sunflower]
A sunflower with core set $C$ and $p$ petals 
 is a collection of sets $S_1, \ldots, S_p$ such that
for all $i \neq j$: $S_i \cap S_j = C$ and $(S_i - C) \cap (S_j - C) = \emptyset$.

\end{definition}

We use the following lemma due to Erd\"os and Rado \cite{er}.
\begin{lemma}[Sunflower Lemma]
Let $\mathcal{F}$ be a family of sets of cardinality $k$ each. 
If $|\mathcal{F}| > k!(p-1)^k$ then $\mathcal{F}$ contains a sunflower
with $p$ petals.
\label{lemma:sunflower}
\end{lemma}

 {\sf Proof of Theorem~\ref{thm:forbidden-H}:} Let $G$ be a graph on $n$ vertices such that every vertex of $G$ is relevant for the property of containing $H$. Let \[\mathcal{F} := \{S \mid |S| = k \ \&\  H \text{ is a subgraph of } G[S]\}.\]
 Since every vertex of $G$ is relevant for $\mathcal{P}_H$, we have: 
 $|\mathcal{F}| \geq n/k$. Now from Lemma~\ref{lemma:sunflower} we can conclude that
 $\mathcal{F}$ contains a sunflower on at least $|\mathcal{F}|^{1/k} / k = \Omega(n^{1/k})$ petals. Let $C$ be the core of this sunflower. We consider the restriction of $\mathcal{P}_H$
 on $G$ where every vertex in $C$ is present. Since $|C| < k$, $G[C]$ does
 not contain $H$. Now it is easy to check that one must query at least one vertex from
 each petal in order to determine $\mathcal{P}_H$.
 
 \hfill$\Box$
 
 Using similar technique we prove Theorem~\ref{thm:lower} (proof in Appendix~\ref{appx:her-lower}) showing that $\minc(\mathcal{P})$ for any hereditary $\mathcal{P}$ can not fall to a constant.

\subsection{Some tight bounds}

We manage to show that Theorem~\ref{thm:upper} is tight for several special properties like Independence, Triangle-freeness, Bounded-degree etc. In Appendix~\ref{appx:tight-bound} we present them in detail. In order to prove the tight bounds, we show several inequalities which might be of independent interest combinatorially.
We present one such inequality below.

\begin{theorem}\label{thm-n-chi}
Let $\mathcal{I}$ denote the property of being an independent subset of nodes (equivalently the property of being an empty graph). Then,

\[\mathcal{G}\text{-}\minc(\mathcal{I}) \geq n/ \chi.\]
where $\chi$ is the maximum chromatic number of a graph $G \in \mathcal{G}$.
\label{thm:n/chi}
\end{theorem}
\begin{proof}  
 Let $G \in \mathcal{G}$ be a graph on $n$ vertices such that every vertex of $G$ is relevant
 for $\mathcal{I}$, i.e., $G$ does not contain any isolated vertices. 
 Consider a coloring of vertices of $G$ with $\chi$ colors. Let $C_i$ denote the set of
 vertices colored with color $i$. We pick a coloring that maximizes $\max_{i \leq \chi}\{|C_i|\}$.
  Let $C_{max}$ denote such a color class with maximum number of vertices in this coloring.  Thus $|C_{\max}| \geq n/ \chi$. We consider the following two cases: 
 
  {\sf Case 1:} $|C_{\max}| \leq (1-\frac{1}{\chi})n$

 In this case, the adversary answers all the vertices in $C_{\max}$ to be present. Since $C_{\max}$ is maximal
  and $G$ does not contain any isolated vertices, every vertex outside $C_{\max}$ must
  be connected to some vertex in $C_{\max}$.  As long as any of these outside vertices are present there will be an edge. Hence we get a lower bound of
   $n - |C_{\max}|\geq n/\chi$.
  
    {\sf Case 2:} $|C_{\max}| > (1-\frac{1}{\chi})n$
    
    Since there are no isolated vertices in $G$,  every vertex in  $C_{\max}$ must have an edge to some vertex in $C_i \neq C_{\max}$. Furthermore as $|C_{\max}| > (1-\frac{1}{\chi})n$, there are totally at least $(1-\frac{1}{\chi})n$ edges incident on $C_{\max}$.  
    
    Now the  vertices outside $C_{\max}$ are colored with $(\chi -1)$ colors. Thus there must exists a $C_i$ such that at least $\frac{(1-\frac{1}{\chi})n}{\chi-1} = n/\chi$ edges incident on $C_{\max}$ are also incident on $C_i$. Now the adversary answers all the vertices in that $C_i$ to be present. Then one must check at least $n/\chi$ vertices from $C_{\max}$ because as soon as any one of them is present we have an edge in the graph.
 \end{proof}

\section{Results on restricted graph classes}
\label{sec:restrictions}

All the proofs of this section are deferred to Appendix~\ref{appx:restrict}.

\subsection{Triangle-freeness in planar graphs}
A graph $G$ is called {\em inherently sparse} if every subgraph of $G$ on $k$
nodes contains $O(k)$ edges.

\begin{theorem}\label{lemma-sparse-tri}
Let $\mathcal{G}_s$ be a family of inherently sparse graphs on $n$ vertices and $\mathcal{T}$ denote the property of being triangle-free. Then,

\[\mathcal{G}_s\text{-}\minc(\mathcal{T}) = \Omega(\sqrt{n}).\]

\end{theorem}

As a consequence  we obtain the same for the class of planar graphs.

\subsection{Acyclicity in planar graphs}
\begin{theorem}\label{planar:acyclic}
Let $\mathcal{G}_{\mathcal{P}_3}$ be a family of $3$-connected planar graphs and $\mathcal{C}$ denote the property of being acyclic. Then,
\[ \mathcal{G}_{\mathcal{P}_3}-\minc(\mathcal{C}) = \Omega(\sqrt n). \]
\label{thm:planar-acyclic}
\end{theorem}
\begin{proof}
Let $G \in \mathcal{G}_{\mathcal{P}_3}$ be a graph on $n$ vertices and $m$ edges such that every vertex is relevant for the acyclicity property.
Let $d_{max}$ denote the maximum degree of $G$.

{\sf Case 1:} $d_{\max} > \sqrt n$: We use the following fact: In a 3-connected planar 
graphs, removing any vertex leaves a facial cycle around it. We apply this for the maximum degree vertex. In other words, we have a (not necessarily induced) wheel with $d_{max}$ spokes (some spokes might be missing). See Figure \ref{fig:wheel}. The adversary answers the central vertex
of the wheel to be present. We can find a matching of size $\Omega(n)$ among the vertices of the cycle. Hence we have $\Omega(n)$ sensitive blocks
of length $2$ each, which can not be left un-queried. 

\begin{center}
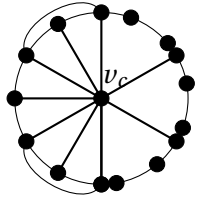
\begin{figure}[H]
 \begin{center}
\begin{tikzpicture}[scale=0.3]

 \tikzstyle{nnode} = [circle, draw=black, fill=black, inner sep=0pt, minimum size = 2mm]
 \tikzstyle{txt} = [above, text width=2cm, text centered]

         \node () [nnode] at (0,0){};
         \node () [txt] at (0.6,0.1){$v_c$};
         \draw (0, 0) circle (3.8cm);
        \foreach \x  in {0, 30, ..., 180}{
        \node () [nnode] at (\x+90:3.8){};
         \draw[thick] (\x+90:3.8)--(0,0);
         };
         
       \foreach \x  in {190, 220, ..., 360}{
        \node () [nnode] at (\x+90:3.8){};
        };
       \foreach \x  in {180, 240, ..., 330}{
        \node () [nnode] at (\x+90:3.8){};
         \draw[thick] (\x+90:3.8)--(0,0);
         };
  \path[every node/.style={font=\sffamily\small}]
     (150:3.8) edge[bend left = 90] node [right] {} (90:3.8)  
     (270:3.8) edge[bend left = 90] node [right] {} (210:3.8);

\end{tikzpicture}
\end{center}
\caption{A wheel with $d_{max}$ spokes}\label{fig:wheel}
\end{figure}
\end{center}

{\sf Case 2:} $d_{\max} \leq \sqrt n$: We use the fact that every 3-connected graph
must have at least $3n/2$ edges. Now using Lemma~\ref{lem:d-max} we obtain 
a lower bound of $(m - n)/d_{max} \geq \Omega(\sqrt n)$.

\end{proof}

We can generalize the above proof to any planar graph (See Appendix~\ref{planar:acyclic-gen}). 

\section{Conclusion \& open directions}
\label{sec:open}

\begin{itemize}

\item {\sf \bf Weak-evasiveness in the presence of symmetry:}
Is it true that every hereditary graph property $\mathcal{P}$ in the node-query setting is weakly-evasive under symmetry, 
i.e., 
$\mathcal{G}_{\mathcal{T}}\text{-}\minc(\mathcal{P}) = \Omega(n)$? 
What about the randomized case?

\item {\sf \bf Polynomial lower bound in the absence of symmetry:}
How low can \\$\minc(\mathcal{P})$ go for a hereditary $\mathcal{P}$ 
in the absence of symmetry? Is it possible to improve the $\Omega(\log n/\log \log n)$ bound substantially?

\item {\sf \bf Further restrictions on graphs:}
How low can $\mathcal{G}\text{-}\minc(\mathcal{P})$ go for hereditary properties $\mathcal{P}$ on restricted classes of graphs $\mathcal{G}$ such as social-network graphs, planar graphs, bipartite graphs, bounded degree graphs etc? 

\item {\sf \bf Tight bounds on $\minc$:}
What are the tight bounds for natural properties such as acyclicity, planarity, containing a cycle of length $t$, path of length $t$?

\item {\sf \bf Extension to hypergraphs:}
What happens for hereditary properties of (say) $3$-uniform hypergraphs in node-query setting? We note that $\minc(\mathcal{I}) = \Theta(n^{1/3})$ for $3$-uniform hypergraphs.
What about other properties?

\item {\sf \bf Global vs local:}
We note (See Appendix~\ref{connectivity}) that global connectivity requires $\Theta(n)$ queries whereas
the cost of $s$-$t$ connectivity for fixed $s$ and $t$ can go as low as $O(1)$. 
What about other properties such as min-cut?

\item {\sf \bf How about $\maxc$ upper bounds? :}
From algorithmic point of view, it might be interesting to obtain good upper bounds on the $\maxc(\mathcal{P})$ for some natural properties. It might also be interesting
to investigate $\mathcal{G}\text{-}\maxc(\mathcal{P})$ for several restricted graph classes such as social-network graphs, planar graphs, bipartite graphs etc.

\end{itemize}

\bibliographystyle{plain}
\bibliography{bibo}

\newpage

\appendix{\bf \Large Appendix}
\label{apx}

\section{Presence of symmetry}
\label{appx:symmetry}

\begin{lemma}
\label{appx:symmetry-planar}
Let $G$ be a graph on $n$ vertices,  $m$ edges,  and maximum degree $d_{max}$. Let $\mathcal{P}'$ denote the property of being planar. Then,
\[cost(\mathcal{P}', G) \geq (m-3n+6)/d_{max}.\]
\end{lemma}
\begin{proof} To make $G$ planar one has to remove at least $(m-3n+6)$ edges from the graph $G$. Removing one vertex can remove at most $d_{max}$ edges.
Thus the size of minimum hitting set of $G$ is at least $(m-3n+6)/d_{\max}$. The adversary answers all vertices outside this minimum hitting set
to be present. Now the algorithm must query every vertex in the minimum hitting set.
\end{proof}

\begin{theorem} 
\label{thm:appx:symmetry-planar}
 \[ \mathcal{G}_{\mathcal{T}}\text{-}\minc(\mathcal{P}') = \Omega(n).\]

\end{theorem}
\begin{proof}
Since $G$ is transitive, $G$ is $d$ regular for some $d$ \cite{agt}. Therefore $m = dn/2$ and $d_{max} = d$. Hence for $d \geq 7$ using Lemma~\ref{appx:symmetry-planar} we get the desired bound.
\end{proof}

\section{Hereditary graph property lower bound}\label{appx:her-lower}

\begin{reptheorem}{thm:any-her-lower}
(Restated) For any hereditary graph property $\mathcal{P}$
\[ \minc(\mathcal{P}) = \Omega \left( \frac{\log n}{\log \log n}\right).\]

\end{reptheorem}
 \begin{proof} 
 
 Let $G$ be a graph on $n$ vertices such that every vertex of $G$ is relevant for $\mathcal{P}$. Let $k$ be the largest integer such that $G$ contains a minimal forbidden subgraph for $\mathcal{P}$ on $k$ vertices.
 
 {\sf Case 1:} $k \geq \frac{\log n}{2\log \log n}$.
 
 Since one must query all the vertices of a minimal forbidden subgraph, we obtain
 a lower bound of $k = \Omega(\log n/\log \log n)$.

{\sf Case 2:} $k < \frac{\log n}{2\log \log n}$.

 Since every vertex of $G$ is relevant for $\mathcal{P}$ and all the minimal forbidden subgraphs of $\mathcal{P}$ present in $G$ are of size at most $k$, every vertex of $G$ must belong to some minimal forbidden subgraph of size at most $k$. 
Consider the property $\mathcal{P}_k$ obtained from $\mathcal{P}$ by omitting
the minimal forbidden subgraphs of $\mathcal{P}$ on $k$ or more vertices.
Our simple but crucial observation is that $\mathcal{P}$ and $\mathcal{P}_k$ 
are equivalent as far as $G$ is concerned. Therefore, they have the same complexity.
Now we define
$\mathcal{F}_i$ for $i \leq k$ as follows:
\[\mathcal{F}_i := \{S \mid |S| = i \ \&\  G[S] \notin \mathcal{P}\ \&\  
\forall T \subset S: G[T] \in \mathcal{P}\}.\]

 Since every vertex of $G$ is relevant for $\mathcal{P} \equiv \mathcal{P}_k$, we have: 
 $|\bigcup_{i=1}^k \mathcal{F}_i| \geq n/k$.
 Since $\mathcal{F}_i$ and $\mathcal{F}_j$ are disjoint when $i \neq j$, we have
 $\sum_{i=1}^k |\mathcal{F}_i| \geq n/k$. Therefore one of the $\mathcal{F}_i$s must be
 of size at least $n/k^2$. We denote that $\mathcal{F}_i$ by $\mathcal{F}'$.

  Now from Lemma~\ref{lemma:sunflower} we can conclude that
 $\mathcal{F}'$ contains a sunflower on at least $|\mathcal{F}'|^{1/k} / k$ petals. Let $C$ be the core of this sunflower. We consider the restriction of $\mathcal{P}$
 on $G$ where every vertex in $C$ is present. Since $|C| < i$, by definition of $\mathcal{F}_i$ we must have $G[C] \in \mathcal{P}$. Now it is easy to check that one must query at least one vertex from
 each petal in order to determine $\mathcal{P}$. A simple calculation yields that
 one can obtain a lower bound of $\min\{k, \frac{2^{\Omega(\log n/k)}}{k}\}$.
 When $k = \log n/ (2 \log \log n)$, this gives us $\Omega(\log n / \log \log n)$ bound.

 \end{proof}

 \section{Proof of tight bounds}
 \label{appx:tight-bound}
 
In this section we show that Theorem~\ref{thm:upper} is tight for several
special properties.
\subsection*{Independence/Emptiness}
\label{sec:ind}

 \begin{theorem}
 Let $\mathcal{I}$ denote the property of being an independent subset of nodes (equivalently the property of being an empty graph). Then,
 \[\minc(\mathcal{I}) = \Omega(\sqrt n).\]
 \label{thm:ind-lower}
 \end{theorem}

   \begin{proof}
 
 Let $G$ be a graph without isolated vertices. If $\chi(G) \leq \sqrt n$ then
 from Theorem~\ref{thm:n/chi} we get $\Omega(\sqrt n)$ lower bound.
 Otherwise $G$ must have a vertex of degree $\Omega(\sqrt n)$. In this case the adversary answers this vertex to be present. Now the algorithm must query all its neighbours.
 \end{proof}

\subsection*{Triangle-freeness}

\begin{theorem} Let $\mathcal{T}$ denote the property of being triangle-free. Then,
\[ \minc(\mathcal{T}) = \Omega(n^{1/3}). \]
\label{thm:triangle-free}
\end{theorem}

\begin{center}
\begin{figure}[H]
\begin{center}
\begin{tikzpicture}[scale=0.6]

 \tikzstyle{nnode} = [circle, draw=black, fill=black, inner sep=0pt, minimum size = 2mm]
 \tikzstyle{txt} = [above, text width=2cm, text centered]

\fill [rounded corners, gray!20] (4,-3) rectangle (6,3);

\draw (0,0) ellipse (1cm and 2cm);
\node (v11) [nnode] at (0,0.8){};
\node (v12) [nnode] at (0,0.4){};
\node (v1) [nnode] at (0,-0.8){};

\draw (v11)--(v12);

\foreach \y in {2.5,2,1.5}{
\node () [nnode] at (5,\y){};
\draw (v11)--(5,\y);
\draw (v12)--(5,\y);
};

\foreach \y in {-1,-1.5,-2,-2.5}{
\node (v11) [nnode] at (5,\y){};
\draw (v1)--(5,\y);
};

\draw (5,-2.5)--(5,-2);
 \draw (5,-1.5)--(5,-1);

\node () [txt] at (0,-0.4){$\vdots$};
\node () [txt] at (5,0){$\vdots$};
\node () [txt] at (0,2){$S$};
\node () [txt] at (3.5,2.5){$d_{max}^2$};
\node () [txt] at (3.5,-2.5){$d_{max}^1$};

\end{tikzpicture}
\end{center}
\caption{Tight lower bound for triangle-freeness}\label{fig:gen-tri}
\end{figure}
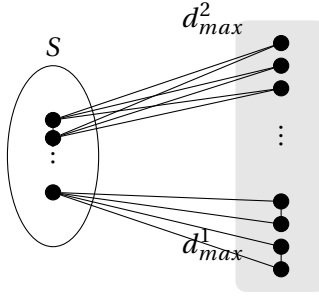
\end{center}

\begin{proof}
Let $G$ be a graph such that every vertex belongs to some triangle.
Let $S$ denote a minimal hitting-set for all triangles, i.e., every triangle must share
a vertex with $S$. Let $d^1_{max}$ denote the maximum number of triangles supported at
a vertex in $G$, i.e., maximum number of triangles whose common intersection is that vertex. Similarly let $d^2_{max}$ denote the maximum number of triangles supported at
an edge in $G$. We consider the following cases:

{\sf Case 1:} $d^1_{max} \geq \Omega(n^{2/3})$: The adversary can answer the common vertex to be present. Now consider the graph induced on the remaining endpoints of each the $d^1_{max}$ triangles. Note that this graph does not contain any isolated vertices. Also note that this graph contains an edge if and only if we have a triangle in the original graph with the common vertex. Hence from Theorem~\ref{thm:ind-lower} we
get a bound of $\sqrt{d^1_{max}} = \Omega(n^{1/3})$. 

{\sf Case 2:} $d^2_{max} \geq \Omega(n^{1/3})$: The adversary can answer both the endpoints of the common edge to be present, which would force the algorithm to query each of the remaining $d^2_{max}$ vertices. This gives $\Omega(n^{1/3})$ bound.

{\sf Case 3:} $|S| = \Omega(n^{1/3})$. 
The adversary answers all the vertices outside the hitting set to be present.
Now as soon any of the vertex in $S$ is present we have a triangle. This gives again
$\Omega(n^{1/3})$ bound.

Finally: at least one of the three cases above must happen, otherwise
there will be some vertex in $G$, which will not be contained in any triangle. 

\end{proof}

\subsection*{Containing path of length $t$, $P_t$ and cycle of size $t$, $C_t$}

\begin{theorem}
\footnote{Proof of this theorem will appear in the final version.} Let $\mathcal{P}_t$ denote the property of containing a path of length $t$, and let $\mathcal{C}_t$ denote the property of containing a cycle of size $t$. Then,
\[ \minc(\mathcal{P}_t) = \Omega(\sqrt n). \] and
\[ \minc(\mathcal{C}_t) = \Omega(n^{1/3}). \]

\label{thm:p-t}
\label{thm:c-t}
\end{theorem}

\subsection*{Bounded degree}
\label{sec:bounded-degree}

\begin{theorem}
Let $\mathcal{B}_d$ denote the property of having maximum degree at most $d$, where $d$ is a constant. Then,
\[ \minc(\mathcal{B}_d) = \Omega(\sqrt n). \] 
\label{thm:bounded-degree}

\end{theorem}

\begin{center}
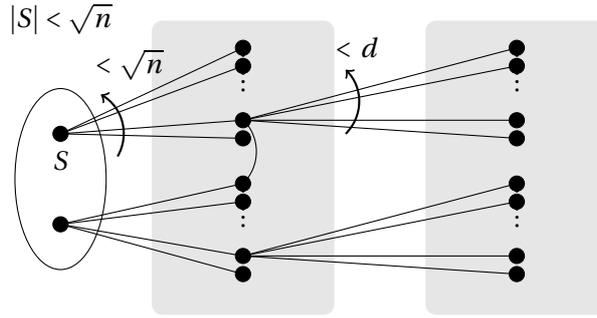
\begin{figure}[H]
 \begin{center}
\begin{tikzpicture}[scale=0.6]

 \tikzstyle{nnode} = [circle, draw=black, fill=black, inner sep=0pt, minimum size = 2mm]
 \tikzstyle{txt} = [above, text width=2cm, text centered]

\fill [rounded corners, gray!20] (2,-3) rectangle (6,3.5);
\fill [rounded corners, gray!20] (8,-3) rectangle (12,3.5);

\draw (0,0) ellipse (1cm and 2cm);
\node (v11) [nnode] at (0,1){};
\node (v12) [nnode] at (0,-1){};
\node (s) [txt] at (0,0){$S$};
\node (scard) [txt] at (0,3){$|S| < \sqrt{n}$};

\foreach \y in {2.9,2.5,1.3, 0.9}{
\node () [nnode] at (4,\y){};
\draw (0,1)--(4,\y);
};

\node () [txt] at (4,2.1){$\vdots$};
\node () [txt] at (4,1.7){$\vdots$};

\foreach \y in {-0.1,-0.5,-1.7,-2.1}{
\node () [nnode] at (4,\y){};
\draw (0,-1)--(4,\y);
};

\node () [txt] at (4,-0.9){$\vdots$};
\node () [txt] at (4,-1.3){$\vdots$};

\foreach \y in {2.9,2.5,1.3, 0.9}{
\node () [nnode] at (10,\y){};
\draw (4,1.3)--(10,\y);
};

\node () [txt] at (10,2.1){$\vdots$};
\node () [txt] at (10,1.7){$\vdots$};

\foreach \y in {-0.1,-0.5,-1.7,-2.1}{
\node () [nnode] at (10,\y){};
\draw (4,-1.7)--(10,\y);
};

\node () [txt] at (10,-0.9){$\vdots$};
\node () [txt] at (10,-1.3){$\vdots$};

\draw[thick,->] (1,0.5) +(0:0.25cm) arc [radius=1cm,start angle=-30,end angle=60];
\node () [txt] at (1.5,2) {$< \sqrt{n}$};

\draw[thick,->] (6,1) +(0:0.25cm) arc [radius=1cm,start angle=-45,end angle=45];
\node () [txt] at (6.5,2.5) {$< d$};

\draw (4,-0.1) arc [radius=1cm,start angle=-45,end angle=45];
\end{tikzpicture}
\end{center}
\caption{Tight lower bound for bounded degree}\label{fig:star}
\end{figure}
\end{center}

\begin{proof} 
Let $G$ be a graph such that every vertex belongs to some $d$-star ($d$ vertices incident on a single vertex). Let $S$ denote the hitting set for all stars of size $d$ in $G$. Let $d_{max}$ denote the maximum degree of $G$.

{\sf Case 1:} $d_{max} \geq \sqrt n / 10d$: The adversary answers the maximum degree vertex to be present. Now one must query all $\Omega(d_{max})$ of its neighbors.

{\sf Case 2:} $|S| \geq \sqrt n / 10 d$: The adversary answers all vertices outside the hitting set to be present. One must query the entire hitting set.

Finally: we claim that one of the above two cases must happen. Otherwise we have 
at most $n/100d^2$ neighbors of the hitting set, each of them can have at most $d-1$ other neighbours. This leaves some vertex $t$ not in $d$-star.
\end{proof}

In fact the proof above generalizes to {\em local properties}
 of bounded degree graphs. 
 For a property $\mathcal{P}$ we define $\mathcal{P}_L$ as follows:
 \begin{definition}[Local Property]
 $G$ satisfies $\mathcal{P}_L$ if and only if for every vertex of $G$ the graph induced on its neighbors satisfies $\mathcal{P}$.
\end{definition}
For instance: bipartite graphs are locally acyclic.
It turns out that $\mathcal{P}_L$ is hereditary for any hereditary $\mathcal{P}$.
Moreover, every graph in the forbidden family $\mathcal{F}_{\mathcal{P}_L}$ has a universal vertex, i.e., 
a vertex adjacent to all other vertices.

\begin{theorem}
For any hereditary $\mathcal{P}$ 
\[\minc(\mathcal{P}_L \wedge \mathcal{B}_d) = \Omega(\sqrt n). \]
\label{thm:local}
\end{theorem}
\begin{proof}
Let $G$ be a graph such that every vertex belongs to some $d$-star ($d$ vertices incident on a single vertex)
or some $H \in \mathcal{F}_{\mathcal{P}_L}$.
Let $S$ denote the hitting set for $\mathcal{F}_{\mathcal{P}_L} \cup \{S_d\}$ in $G$. Let $d_{max}$ denote the maximum degree of $G$.

{\sf Case 1:} $d_{max} \geq \sqrt n / 10d$: The adversary answers the maximum degree vertex to be present. Now one must query all $\Omega(d_{max})$ of its neighbors (since
there are $d_{max}/d$ disjoint blocks).

{\sf Case 2:} $|S| \geq \sqrt n / 10 d$: The adversary answers all vertices outside the hitting set to be present. One must query the entire hitting set.

Finally: we claim that one of the above two cases must happen.
Since every vertex belongs to some $d$-star
or some $H \in \mathcal{F}_{\mathcal{P}_L}$ and every $H \in \mathcal{F}_{\mathcal{P}_L}$ has a universal vertex, we have that
 every vertex in $G$ is reachable
to some vertex in $S$ by a path of length at most $2$.
 Otherwise we have 
at most $n/100d^2$ neighbors of the hitting set, each of which can have at most $d-1$ other neighbours. This leaves some vertex $t$ not in $\mathcal{F}_{\mathcal{P}_L} \cup \{S_d\}$. 
\end{proof}

\section{Results on restricted graph classes}\label{appx:restrict}

\subsection*{Independence/Emptiness in planar graph}

\begin{theorem}
Let $\mathcal{G}_{\mathcal{P}}$ be a family of planar graphs on $n$ vertices and $\mathcal{I}$ denote the property of being independent. Then,
\[ \mathcal{G}_{\mathcal{P}}\text{-}\minc(\mathcal{I}) = \Omega(n). \]
\end{theorem}

\begin{proof}
As planar graphs are $4$ colorable using Theorem~\ref{thm-n-chi} we can directly conclude this theorem.
\end{proof}

\subsection*{Triangle-freeness in planar graphs}
A graph $G$ is called {\em inherently sparse} if every subgraph of $G$ on $k$
nodes contains $O(k)$ edges.

\begin{reptheorem}{lemma-sparse-tri}
(Restated) Let $\mathcal{G}_s$ be a family of inherently sparse graphs on $n$ vertices and $\mathcal{T}$ denote the property of being triangle free. Then

\[\mathcal{G}_s\text{-}\minc(\mathcal{T}) = \Omega(\sqrt{n}).\]

\end{reptheorem}
\begin{proof}
Let $G =(V,E)\in \mathcal{G}_s$ be a graph on $n$ vertices such that every vertex in $G$ is part of at least one triangle in $G$. 
Let $S \subseteq V$ be a minimal hitting set for triangles in $G$. 
We consider following two cases:

  {\sf Case 1:} $|S| \geq \sqrt{n}$
  
  The adversary answers all the vertices outside $S$ to be present.  
  Hence one has to check all vertices in $S$.
  
    {\sf Case 2:} $|S| < \sqrt{n}$
    
    Since $G$ is inherently sparse, there can be at most $O(|S|) = O(\sqrt{n})$ edges within $S$. Moreover each triangle in $G$ must share either one vertex in $S$ or one edge inside $S$ (it could even be that the whole triangle is inside $S$).
Note that there are $\Omega(n)$ vertices outside $S$ and each of them must belong to some triangle.
 This implies that either there is at least a vertex $v \in S$ or at least an edge $ \{u,v\} \in {S \choose 2}$ s.t. $v$ or $\{u,v\}$ supports at least $\Omega(\sqrt{n})$ triangles outside $S$. Otherwise all triangles in $G$ are not covered by $S$. We consider the following two cases:

    {\sf Case 2a:} an edge $\{u,v\} \in {S \choose 2}$ supports $\Omega(\sqrt{n})$ triangles in $V-S$.

\begin{center}
\begin{figure}[H]
 \begin{center}
\begin{tikzpicture}[scale=0.6]

 \tikzstyle{nnode} = [circle, draw=black, fill=black, inner sep=0pt, minimum size = 2mm]
 \tikzstyle{txt} = [above, text width=2cm, text centered]

\fill [rounded corners, gray!20] (4,-3) rectangle (6,3);

\draw (0,0) ellipse (1cm and 2cm);
\node (v11) [nnode] at (0,0.25){};
\node (v12) [nnode] at (0,-0.25){};
\draw (v11)--(v12);

\foreach \y in {2.5,2,1.5,-1.5,-2,-2.5}{
\node (v11) [nnode] at (5,\y){};
\draw (0,0.25)--(5,\y);
\draw (0,-0.25)--(5,\y);
};

\node (v11) [txt] at (5,0){$\vdots$};
\node (v11) [txt] at (0,2){$S$};
\node (v11) [txt] at (7,0){$\Omega(\sqrt{n})$};

\end{tikzpicture}
\end{center}
\caption{Case 2a}\label{fig-tri-free-2}
\end{figure}
\end{center}
        The adversary makes the edge $\{u,v\}$ present. One then has to query $\Omega(\sqrt{n})$ end points of all $\Omega(\sqrt{n})$ triangles (See Figure \ref{fig-tri-free-2}).

    {\sf Case 2b:} a vertex $v \in S$ supports $\Omega(\sqrt{n})$ triangles.

\begin{center}
\begin{figure}[H]
 \begin{center}
\begin{tikzpicture}[scale=0.6]

 \tikzstyle{nnode} = [circle, draw=black, fill=black, inner sep=0pt, minimum size = 2mm]
 \tikzstyle{txt} = [above, text width=2cm, text centered]

\fill [rounded corners, gray!20] (4,-3) rectangle (6,3);

\draw (0,0) ellipse (1cm and 2cm);
\node (v11) [nnode] at (0,0){};

\foreach \y in {2.5,2,1.5,1, -1,-1.5,-2,-2.5}{
\node (v11) [nnode] at (5,\y){};
\draw (0,0)--(5,\y);
};

\draw (5,2.5)--(5,2);
\draw (5,1.5)--(5,1);
\draw (5,-1)--(5,-1.5);
\draw (5,-2)--(5,-2.5);

\node (v11) [txt] at (5,0){$\vdots$};
\node (v11) [nnode] at (5.5,0.7){};
\draw (5,1.5)--(5.5,0.7);
\draw (5.5,0.7)--(0,0);

\node (v11) [txt] at (0,2){$S$};
\node (v11) [txt] at (7,0){$\Omega(\sqrt{n})$};

\end{tikzpicture}
\end{center}
\caption{Case 2b}\label{fig-tri-free-1}
\end{figure}
\end{center}

Adversary makes the vertex $v$ present. Then the problem reduces to finding an edge in the graph induced on neighbors of $v$. This graph has $\Omega(\sqrt{n})$ non-isolated vertices.  (See Figure \ref{fig-tri-free-1}). As the $G$ is inherently sparse 
the chromatic number of $G$ is constant.\footnote{Pick a smallest degree vertex. Recursively color the rest of the graph with constant colors. Use a different color for the picked vertex.} Now we use Theorem \ref{thm-n-chi} to obtain the $\Omega(\sqrt n)$ bound.
\end{proof}

As a consequence of Lemma~\ref{lemma-sparse-tri} we get the following:

\begin{corollary}\label{cor-pla-tri}
Let 
$\mathcal{G}_{\mathcal{P}}$
be a family of planar graphs on $n$ vertices and ${\mathcal{T}}$ denote the property of being triangle-free. Then,

\[ 
\mathcal{G}_{\mathcal{P}}\text{-}\minc({\mathcal{T}}) = \Omega(\sqrt{n}).\]
\end{corollary}

\subsection*{Acyclicity in planar graphs}
\begin{theorem}\label{planar:acyclic-gen}
Let $\mathcal{G}_{\mathcal{P}}$ be a family of planar graphs on $n$ vertices and let $\mathcal{C}$ denote the property of being acyclic.Then,
\[ \mathcal{G}_{\mathcal{P}}\text{-}\minc(\mathcal{C}) = \Omega(n^{1/16}). \]

\end{theorem}

\begin{proof}
By Theorem~\ref{planar:acyclic} we get $\Omega(\sqrt n)$ bound for testing acyclicity in $3$-connected planar graph.
Then by using Proposition~\ref{prop:degree} and Proposition~\ref{prop:indset}, along with Claim~\ref{clm:exh} we get the desired bound of $\Omega(n^{1/16})$.

\end{proof}

\begin{figure}[H]
 \begin{center}
 
\begin{tikzpicture}[scale=0.5]

 \tikzstyle{nnode} = [circle, draw=black, fill=black, inner sep=0pt, minimum size = 2mm]
 \tikzstyle{txt} = [above, text width=2cm, text centered]

         \node () [nnode] at (0,0){};
         \node () [txt] at (0.2,0.1){$v_c$};
         \draw (0, 0) circle (3.8cm);
        \foreach \x  in {0, 30, ..., 180}{
        \node () [nnode] at (\x+90:3.8){};
         \draw[thick] (\x+90:3.8)--(0,0);
         };
         
       \foreach \x  in {190, 220, ..., 360}{
        \node () [nnode] at (\x+90:3.8){};
        };
       \foreach \x  in {180, 240, ..., 330}{
        \node () [nnode] at (\x+90:3.8){};
         \draw[thick] (\x+90:3.8)--(0,0);
         };
  \path[every node/.style={font=\sffamily\small}]
     (150:3.8) edge[bend left = 90] node [right] {} (90:3.8)  
     (270:3.8) edge[bend left = 90] node [right] {} (210:3.8);

\end{tikzpicture}
\end{center}
\caption{A wheel with $d_{max}$ spokes, $|K_c|$ cycle vertices and some chords}
\label{fig:case2a}
\end{figure}
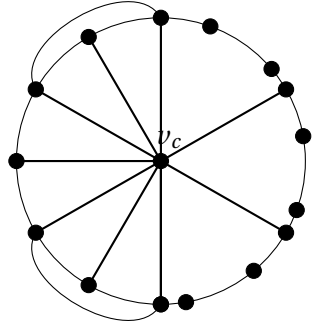

 In Theorem~\ref{planar:acyclic} we get $\Omega(\sqrt n)$ bound for testing acyclicity in $3$-connected planar graph. 
 Now we move to a lower bound for deciding acyclicity in $2$-connected and $1$-connected
 planar graphs. Our main tool here are the incidence graphs of triconnected components
 $G_I = (V_I, E_I)$ of the given $2$-connected/connected planar graph $G$ and 
 and separating sets which are vertices removing which disconnects two triconnected components.
 This is a graph whose vertices are the triconnected components of $G$, and each vertex is joined by
 an edge if they have a separating vertex (or pair). Also define the incidence graph of the 
 separating sets of a triconnected component, $G_S = (V_S, E_S)$, which is the graph where 
 the vertices are the separating sets of a triconnected component and we have an edge between
 two vertices if the corresponding separating sets intersect.
 We will distinguish between the case of separating sets of size $2,1$ and $0$. Graphs which have
 separating sets of size $2$ are precisely the biconnected graphs, with size $1$ are connected
 and size $0$ are disjoint triconnected components.  If our graph $G'$ has a large triconnected component of
 size greater than $\sqrt{n}$, then from Theorem~\ref{planar:acyclic}, we already get a lower bound 
 of $n^{1/4}$. When all the $3$-connected components are small ($< \sqrt{n}$ in size) there 
 must be at least $\sqrt{n}$ of them -- hence $|V_I| \geq \sqrt{n}$. We prove the following claim:

 \begin{claim}\label{clm:exh}
 In the incidence graph of the triconnected components, either there is a component of large 
 degree ($ >\sqrt{n}$) or there is an independent set of size at least $\sqrt{n}$.
 \end{claim}
 \begin{proof}
 The degree and the size of the independent set of the incidence graph of triconnected 
 components cannot be simultaneously small ($< \sqrt{n}$) because otherwise there will be
 isolated triconnected components. 
 \end{proof}
 
 \begin{proposition} \label{prop:degree}
  In a $2$-connected/ $1$-connected graph $G$, if $G_I$ has a vertex $t \in V_I$ of large degree, then
  acyclicity testing requires $n^{1/16}$ queries.
 \end{proposition}

 \begin{proof}
  When a triconnected component has a large degree ($ d > \sqrt{|V_I|} = \Omega(n^{1/4})$
  in its incidence graph): Let $t \in V_I$ be the triconnected component with largest degree, namely $d$,
  and let $t_1, t_2, \cdots, t_d$ be its neighbours. Now notice that in the incidence graph $G_S$, either
  there are more than $\sqrt{|V_S|}$ disjoint connected components, in which case, by Proposition~\ref{prop:indset}
  a lower bound is at least $\sqrt{|V_S|} = n^{1/8}$. Else there is one large connected component of 
  separating sets, which is of size at least $\sqrt{|V_S|} = n^{1/8}$ and in this component, both the degree
  and independence number cannot be simultaneously small. First consider the case when the degree is large, i.e.
  at least $n^{1/16}$ (which is at least square root of the number of vertices in the component). In this case,
  we can fix an induced cycle passing through the separating set vertices in each one of the $t_i$'s corresponding 
  to the separating sets and 
  and answer yes for queries made to vertices in the cycle and no for queries made to vertices outside this cycle.
  Since at least one vertex in each $t_i$ has to be queried, this gives a lower bound of $n^{1/16}$. In the case
  when the independence number of this component is large (at least $n^{1/16}$), the separating sets of the corresponding
  $t_i$'s are independent and hence we can use the same strategy -- fix an induced cycle in these $t_i$, and 
  answer yes for queries made to vertices in the cycle and no for queries made outside
  of this cycle. This gives a lower bound of $n^{1/16}$.
  
\end{proof}

\begin{proposition} \label{prop:indset}
 In a $2$-connected/ $1$-connected graph $G$, if $G_I$ has an independent set of large size, then
  acyclicity testing requires $n^{1/4}$ queries.
\end{proposition}

\begin{proof}
 When $G_I$ has a large independent set ($|V_{I}^{\mbox{ind}}| > \sqrt{|V_I|}$):
 The adversary fixes a face in every
 triconnected component and answers no on all queries made to vertices except those that lie on
 this face. It is clear that at least one vertex is queried in a component, and this gives a lower
 bound of $|V_{I}^{\mbox{ind}}|  \geq \sqrt{|V_I|} \geq n^{1/4}$.   
\end{proof}

\section{Global vs local connectivity}\label{connectivity}
\begin{theorem}
\label{thm:gl-loc}
Let $Global\text{-}Con$ denote the problem of testing whether a graph is connected or not and
let $Local\text{-}Con$ denote the problem testing given two specified vertices $s$ and $t$ whether there is a
path between $s$ and $t$. Then,\\
 (a) $\minc(Local\text{-}Con) = \Theta(1)$ whereas \\
 (b) $\minc(Global\text{-}Con) = \Theta(n)$.

\end{theorem}

\begin{proof}
This theorem follows directly from 
Lemma \ref{s-t-conn-upper} and Lemma \ref{global-conn-lower}. 
\end{proof}

\begin{lemma}\label{s-t-conn-upper}
Let $Local\text{-}Con$ denote the problem of testing given two specified vertices $s$ and $t$ whether there is a
path between $s$ and $t$. Then,
\[\minc(Local\text{-}Con) = O(1).\]
\end{lemma}
\begin{proof}
Consider the star graph on $n$ vertices, i.e., a vertex $v$ connected to $n-1$ vertices $u_1, \ldots, u_{n-1}$.
It is easy to check that for any $s$ and $t$ we have to check at most one more vertex other than $s$ and $t$ to determine if there is
a path between $s$ and $t$.
\end{proof}

\begin{lemma}\label{global-conn-lower}
Let $Global\text{-}Con$ denote the problem of testing whether a graph is connected or not. Then,
\[\minc(Global\text{-}Con) = \Omega(n).\]
\end{lemma}
\begin{proof}
The convention we use is that the singleton vertex is connected and graph on $0$ nodes is connected.
Let $G$ be a graph on $n$ vertices. We are interested in global connectivity of subgraphs of $G$. 
If the complement of $G$ has a matching of size $\Omega(n)$ then each matching edge is a sensitive block with respect to empty graph, i.e., the graph 
with just these two nodes present is not connected. Hence we get the desired lower bound.
If on the other hand the maximum matching size in the complement of $G$ is at most say $n/4$ (hence  the vertex cover size at most $n/2$), then
$G$ must contain a clique on at least $n/2$ vertices. 

Now we consider vertices outside this clique. 
Since the connectivity property is non-trivial on $G$, at least one such vertex say $v$ must exist.
If $v$ has at most $n/4$ edges to the vertices of the clique, we answer this vertex to be present and its neighbors in clique to be absent. 
Now as soon as any of the remaining $n/4$ vertices  from the clique are present, we get a disconnected graph.
If $v$ has at least $n/4$ neighbors in the clique, we take a non-neighbor say $u$ from the clique. Now $u$ and $v$ have at least $n/4$ common neighbors.
We make $u$ and $v$ to be present. As soon as any of their common neighbors are present the graph is connected.
This gives us $\Omega(n)$ bound.
\end{proof}

\section*{Perfect matching}

\begin{theorem}\footnote{Proof of this theorem is similar to the proof of Theorem~\ref{thm:gl-loc} and will appear in the final version.}
\label{perfect-matching}
Let $\mathcal{PM}$ be the property of containing a perfect matching in a graph $G$ on $n$ vertices then,
\[ \minc(\mathcal{PM}) = \Theta(n).\]
\label{thm:perfect-matching}
\end{theorem}

\end{document}